\documentclass[12pt]{article}
\usepackage{comment}
\usepackage{natbib}
\usepackage{fullpage}
\usepackage{times}
\usepackage[normalem]{ulem}
\usepackage{fancyhdr,graphicx,amsmath,amssymb, mathtools, scrextend, titlesec, enumitem}
\usepackage[ruled,vlined]{algorithm2e} 
\newtheorem{theorem}{Theorem}

\newtheorem{proposition}[theorem]{Proposition}
\newtheorem{remark}[theorem]{Remark}

\newenvironment{proof}[1][Proof]{\noindent\textbf{#1.} }{\ \rule{0.5em}{0.5em}}

\title{Perceived Social Norms and Information Disclosure}
\author{Senran Lin\thanks{University of
Edinburgh, United Kingdom: senranlin@outlook.com.}}

\begin{document}
\maketitle

\begin{abstract}
    Social-information messages are widely used to influence norm perceptions and norm-relevant behavior, but their effects may depend on what they report and how they are disclosed. This paper develops a belief-based framework connecting personal values and perceived social norms through inference about a latent appropriateness standard. The framework shows how messages affect perceived social norms, with potential implications for subsequent behavior, through two mechanisms: First, personal-value and perceived-norm statistics encode previous participants' private information differently because perceived-norm reports also incorporate an inference about others' assessments. Second, relative to private disclosure, public disclosure also affects recipients' beliefs about others' assessments, as they expect other group members to update using the same statistic. These mechanisms imply that, for equally sized changes on the same scale, perceived-norm statistics induce larger shifts in current perceived norms than personal-value statistics, and that public disclosure induces larger marginal responses than private disclosure. Thus, the effect of a norm-relevant message can depend not only on the value it reports, but also on how the reported statistic is constructed and who is known to observe it.
\end{abstract}

In experiments and interventions, social-information messages are widely used to influence norm-relevant behavior, and they often report what others did, what they personally regard as appropriate, or what they believe the reference group approves of.\footnote{For example, social-norms messaging interventions use social proof, social comparison, injunctive, and combined messages; see \cite{PapakonstantinouEtAl2025}.} Evidence from a reconciliation radio intervention in Rwanda illustrates the economic relevance of norm-related communication: it changed how people perceived social norms and influenced cooperation and collective decision-making \citep{Paluck2009,PaluckGreen2009}. A growing empirical literature further compares how different message forms shape norm perceptions, expectations, and behavior \citep{BicchieriXiao2009,SchultzNolanCialdini2007,BonanCattaneoDAddaTavoni2020,PapakonstantinouEtAl2025}.

A central question is whether these different forms of information should have the same impact, or whether norm-relevant statistics affect perceived norms and, through them, norm-relevant behavior in systematically different ways. In particular, does the impact of a change in reported information depend on the message form?

To address this question, we draw on a central distinction in the social-norms literature. A substantial body of research distinguishes among personal values, perceived injunctive social norms, and empirical expectations. A personal value captures what an individual herself regards as appropriate; a perceived injunctive social norm captures what she believes others regard as appropriate; and an empirical expectation captures what she believes others will do. Personal values and perceived injunctive social norms thus share appropriateness as their conceptual basis, while differing in whose assessment is considered. Experimental studies have separately elicited these measures and found them to be distinct but related \citep{burkskrupka2012multimethod,dAddaDufwenbergPassarelliTabellini2020,catola2021personal,catola2025norms,bavsic2024personal,barigozzi2025social}. These findings motivate examining how distinct normative assessments can be systematically related and what this relationship implies when they are used to construct social-information messages.

This paper develops a belief-based framework to address these questions. The central insight is that norm perception can be understood as an inference problem.\footnote{For a parallel informational perspective on how information about others' normative judgments and actions affects norm-related beliefs, see \cite{tremewan2021informational}.} Individuals may not know the relevant appropriateness standard in a given situation; instead, they form assessments from imperfect private cues. The same latent standard therefore connects personal values and perceived injunctive social norms, while empirical expectations can be linked to these assessments indirectly through behavior.

This unified framework allows us to compare how messages based on different norm-relevant measures affect perceived social norms. An account of what others personally think is appropriate need not have the same implications for recipients' perceived norms as one of what others believe the group regards as appropriate, even when both report the same numerical value. Likewise, an account of what others did is less direct: actions may reflect perceived norms through their influence on individuals' choices.

More specifically, the framework clarifies two belief-based mechanisms through which norm-relevant messages affect perceived norms. The first concerns what is disclosed: a statistic based on previous participants' personal values or one based on their perceived social norms. Our analysis focuses on the comparison between these two statistics. Both concern appropriateness and can be reported on the same scale. Yet the two statistics encode information about appropriateness differently, as perceived-norm reports involve an additional inference about others' assessments. Consequently, equally sized changes in the reported statistics need not have the same effect on recipients' perceptions of social norms. Observed actions can also be nested in the framework, but only indirectly: interpreting them as information about perceived norms requires accounting for material incentives, constraints, and motives for norm compliance.

The second concerns disclosure mode: holding the reported statistic fixed, does it matter whether the information is disclosed privately to one individual or publicly to the current group? In both cases, the recipient can use the statistic to update her own assessment of appropriateness. Under public disclosure, she also knows that other current-group members receive the same statistic and expects them to update their assessments accordingly. Because perceived norms involve beliefs about others' assessments \citep{Bicchieri2006}, this creates an additional channel through which public disclosure can affect her perceived social norm, even when the reported statistic is held fixed.

This paper contributes by providing a theoretical foundation for studying the form and disclosure mode of norm-relevant messages. Existing work distinguishes personal values, perceived injunctive norms, and empirical expectations \citep{Bicchieri2006,BicchieriChavez2010,KrupkaWeber2013,dAddaDufwenbergPassarelliTabellini2020}. Less is known, however, about the belief-based mechanisms through which the type of statistic reported and its disclosure mode shape perceived social norms.

The framework developed here helps explain these mechanisms and their implications for perceived social norms. Elicited personal values and perceived-norm reports encode appropriateness differently, so equally sized changes in perceived-norm reports induce larger responses in current perceived norms than changes in personal-value reports. Relative to private disclosure, public disclosure introduces an additional updating channel: the recipient expects other members of her reference group to update their assessments using the same disclosed statistic, leading to larger marginal responses to either message form. A further contribution is to link the distinct norm-relevant concepts--personal values, perceived injunctive social norms, and empirical expectations--within a common mathematical structure that enables the comparison of norm-relevant messages based on them. In particular, personal values and perceived social norms are represented as different levels of inference about a common latent appropriateness standard, while empirical expectations are linked to this standard indirectly through behavior. The resulting comparison illustrates why the choice of elicited measure matters when constructing social-information interventions.


\section{Framework}\label{sec:framework}
We consider environments in which each individual chooses a single action. In such environments, what is appropriate can be represented by a scalar benchmark. Appropriateness is the common conceptual basis for personal values and perceived injunctive social norms \citep{Bicchieri2006,bavsic2024personal}. A natural way to formalize this conceptual commonality is to represent appropriateness by a common latent benchmark, about which individuals form assessments. This common benchmark provides the informational link through which individuals can infer others' assessments from their own information and learn about appropriateness from statistics on others' attitudes or behavior. The framework therefore provides a structure for analyzing how individuals form these assessments, how norm-relevant information affects perceived social norms, and how these perceptions shape behavior.

\subsection{Appropriateness and information structure}
A central concept in the literature on injunctive social norms is the appropriateness of behavior in a given situation. To formalize this idea, we represent the relevant \textbf{appropriateness standard} by a scalar $S$, which serves as a situation-specific benchmark against which actions are evaluated.\footnote{The standard need not be feasible; it serves as a benchmark for evaluating behavior.} For example, in a dictator game, $S$ may be interpreted as an appropriate transfer share; in a public goods game, as an appropriate contribution rate; in a tax compliance setting, as an appropriate declaration level; and in an honesty task, as an acceptable lying threshold.

The standard $S$ is not itself a social norm or a personal value. Rather, it is the underlying benchmark about which individuals form assessments. Although social norms are sometimes characterized as commonly known standards of behavior \citep[e.g.,][]{FehrSchurtenberger2018}, individuals may not directly observe what others regard as appropriate. We therefore allow individuals to be uncertain about the benchmark, to form different assessments of it, and to hold heterogeneous beliefs about others' assessments. We model $S$ as an exogenous random variable and assume
\begin{equation*}
    S \sim N(\mu_s,\nu_s),
\end{equation*}
where $\mu_s$ is the prior mean and $\nu_s$ captures uncertainty about the appropriateness benchmark, with the distribution of $S$ being environment-specific.
The variance $\nu_s$ reflects the \textbf{appropriateness uncertainty}: it is higher when the situation admits a more controversial, complex, or weakly defined appropriateness benchmark. This specification captures the idea that individuals may share a common but imprecise prior view of what is appropriate in a given situation.\footnote{For tractability, we adopt a Gaussian prior–signal structure, which yields closed-form posterior beliefs. Similar approaches are commonly used in related noisy-cognition models; see, e.g., \cite{MooreHealy2008,KhawLiWoodford2021,Woodford2020,EnkeGraeber2023,AzeredoDaSilveiraWoodford2019}.} 

To model the informational component of heterogeneity in individuals' assessments, we follow Harsanyi's information-based perspective, attributing differences in beliefs to differences in information. Drawing on the noisy-cognition literature, we represent this information as a private cue about the appropriateness standard. Let $I$ denote a group of individuals, with $|I|\geq 2$. Each individual $i\in I$ does not directly observe the realized standard, but instead receives
\begin{equation*}
    Y_i=S+\epsilon_i,
    \qquad
    \epsilon_i \overset{i.i.d.}{\sim} N(0,\nu_\epsilon),
\end{equation*}
with realized cue $y_i$, where $\epsilon_i$ is independent of $S$.
This formulation captures an imperfect perception of a common latent benchmark, without taking a stand on how the appropriateness standard itself is formed. The realized cue $y_i$ may reflect gut feeling, attention, personal experience, or memory \citep{MooreHealy2008,KhawLiWoodford2021,OpreaVieider2024,Mullainathan2002}. The noise variance $\nu_{\epsilon}$ measures the idiosyncratic imprecision, potentially arising from limited attention, imprecise recall, or noisy encoding of relevant features. Within the framework, differences in private cues thus generate heterogeneity in individuals' assessments of appropriateness.

\subsection{Personal values and perceived social norms}

Under this framework, individual $i$'s \textbf{personal value}, sometimes referred to as a personal norm, captures her own assessment of what action is appropriate given her realized cue $y_i$. It is defined as
\begin{equation*}
    r_i:=\mathbb{E}[S\mid y_i].
\end{equation*}
Although the appropriateness benchmark is exogenous, personal values are endogenous beliefs about that benchmark. 
Given the information structure above, Bayesian updating gives
\begin{equation}\label{eq:value_cue_init}
    r_i =
    \frac{\nu_\epsilon}{\nu_s+\nu_\epsilon}\mu_s
    +
    \frac{\nu_s}{\nu_s+\nu_\epsilon}y_i.
\end{equation}
Within this specification, $r_i$ is a weighted average of the prior mean $\mu_s$ and the private cue $y_i$. The cue receives more weight when appropriateness uncertainty is higher, i.e., when $\nu_s$ is larger, whereas the prior mean receives more weight when private cues are noisier, i.e., when $\nu_\epsilon$ is larger.

We next define the individual's perceived injunctive social norm. Consistent with Bicchieri's expectation-based account of social norms, we interpret this as individual $i$'s belief about what others in the reference group regard as appropriate. In the present environment, we represent others' assessments by their average.\footnote{For a closely related use of mean aggregation, \cite{gachter2013peer} average the effort levels that participants rated as most socially appropriate. More generally, mean aggregation of appropriateness ratings for each action is used in \cite{KrupkaWeber2013,KimbroughVostroknutov2016}.} Let $H_i$ denote individual $i$'s information set.\footnote{We write $H_i=\{(y_i,X)\}$ when individual $i$ observes an additional statistic $X$ alongside her private cue $y_i$.} Formally, individual $i$'s perceived social norm is
\begin{equation*}
    \mathbb{E}[N_{-i}|H_i],
\end{equation*}
where 
\begin{equation*}
    N_{-i} :=
    \mathrm{avg}_{j\in I\setminus \{i\}}\,\mathbb{E}[S \mid H_j]
\end{equation*}
is the average assessment of appropriateness among the other individuals, evaluated at their own information sets.
This differs from the personal value in two ways. First, it concerns others' assessments rather than $i$'s own assessment. Second, whereas the personal value $r_i$ is defined based solely on the initial private cue $y_i$, the perceived social norm depends on the information set $H_i$, which may incorporate subsequently disclosed information. This distinction is central to the analysis of social-information interventions below: when individual $i$ receives additional norm-relevant information, her information set changes, and so may her perceived social norm.

Expanding the definition gives:
\begin{equation}\label{eq:perceived_norm}
     \mathbb{E}[N_{-i}\mid H_i]
    =\mathbb{E}[\mathrm{avg}_{j\in I\setminus \{i\}}\,\mathbb{E}[S\mid H_j]\mid H_i].
\end{equation}
This expression makes explicit that perceived social norms are beliefs about others' assessments of appropriateness. Thus, if $S$ represents an appropriate transfer share, contribution rate, tax declaration level, or acceptable lying threshold, then $N_{-i}$ represents the reference group's social assessment of appropriateness, and $\mathbb{E}[N_{-i}\mid H_i]$ represents $i$'s belief about that assessment. We focus on perceived norms, rather than a directly observed or commonly known social norm, following the norm-perception perspective emphasized by \cite{TankardPaluck2016}.

\paragraph{Minimal-information (MI) benchmark}
As a benchmark, we first consider the minimal-information (MI) case, in which individuals receive no social-information intervention and observe only their own private cues. Thus, $H_i = \{y_i\}$ for all $i\in I$. In this case, 
\begin{equation*}
N_{-i} =
\mathrm{avg}_{j\in I\setminus \{i\}} r_j.
\end{equation*}
With a slight abuse of notation, we write the \emph{initial perceived social norm} $\mathbb{E}[N_{-i} \mid y_i]$, which represents individual $i$'s belief about the average personal value of others.

Even without social-information disclosure, the common appropriateness benchmark generates a systematic relation between personal values and perceived social norms, while allowing the two assessments to differ.
\begin{remark}
        In the MI case, the perceived social norm is given by
\begin{equation}\label{eq:perceived_norm_ini_value}
    \mathbb{E}[N_{-i}\mid y_i]
    =
    \frac{\nu_\epsilon}{\nu_s+\nu_\epsilon}\mu_s
    +
    \frac{\nu_s}{\nu_s+\nu_\epsilon}\mathbb{E}[S\mid y_i].
\end{equation}
    Hence, relative to individual $i$'s personal value $r_i=\mathbb{E}[S\mid y_i]$, the perceived social norm is pulled toward the prior mean. It need not coincide with her personal value, even before any social-information message is disclosed.
\end{remark}
Using the expression for $r_i$, this can also be written as
\begin{equation}\label{eq:perceived_norm_full_MI}
\mathbb{E}[N_{-i}\mid y_i]=
\frac{\nu_\epsilon(\nu_\epsilon+2\nu_s)}{(\nu_s+\nu_\epsilon)^2}\mu_s
+
\left(\frac{\nu_s}{\nu_s+\nu_\epsilon}\right)^2 y_i.
\end{equation}

This relation follows from the different levels of inference. Individual $i$ uses her own assessment, summarized by $r_i$, to infer others' assessments, while accounting for the fact that others form their own assessments from their own private cues. As a result, her perceived social norm places less weight on her private cue and is more strongly pulled toward the prior mean $\mu_s$ than $r_i$.\footnote{A related pattern appears in
\citet{evdokimov2022higher}: in their experiment with private
signals, second-order beliefs are on average closer to the prior than first-order beliefs.} As shown by \eqref{eq:perceived_norm_ini_value}, the weight on the personal value increases with appropriateness uncertainty $\nu_s$ and decreases with private-cue noise $\nu_\epsilon$, so the relationship between the two assessments varies systematically with the information environment.

The appropriateness benchmark thus allows personal values and perceived social norms to be systematically related without coinciding. Their distinction is not merely terminological: the two measures encode an individual's private cue differently, providing the informational basis for comparing messages constructed from these measures. This difference matters when social-information messages disclose different norm-relevant measures.



\subsection{Behavior and empirical expectations}
We then relate perceived social norms to behavior through a simple model in which an individual trades off material payoffs against deviations from her perceived social norm. Each individual chooses a single action $a_i\in A \subseteq \mathbb{R}$. Let $\pi_i(a_i,a_{-i})$ denote individual $i$'s material payoff. We assume that individual $i$ incurs disutility from deviating from her perceived social norm. That is, conditional on information set $H_i$, individual $i$ evaluates action $a_i$ according to the following decision utility function:
\begin{equation}
\begin{aligned}
    U_i(a_i|H_i)
    &=
    \mathbb{E}[\pi_i(a_i,a_{-i})|H_i]
    -\theta \left(a_i-\mathbb{E}[N_{-i}|H_i]\right)^2,
\end{aligned}
\end{equation}
where $\mathbb{E}[N_{-i}\mid H_i]$ is the perceived social norm\footnote{This specification follows the norm-perception perspective that behavior responds to individuals' perceptions of social norms; see \cite{TankardPaluck2016}.} and $\theta>0$ is a homogeneous norm-sensitivity parameter. Deviations from the perceived social norm are costly, with marginal costs increasing in the size of the deviation. Personal values may also enter decision utility directly, as argued by \cite{bavsic2024personal}. Our specification focuses on the perceived social-norm channel to trace how information-induced changes in norm perceptions translate into actions and, in turn, empirical expectations.

Observed actions can therefore be informative about perceived social norms. To make this link explicit, consider a \emph{unit-cost environment} (UCE), in which, holding others' actions fixed, a higher own action carries a constant marginal material cost equal to one.
The UCE benchmark captures the marginal cost structure of norm-relevant contribution and reporting settings, after suitable normalization.\footnote{Examples include linear public-good contributions, dictator-game transfers, and tax-declaration and honesty-reporting tasks with linear monetary incentives.} In this illustrative benchmark, we abstract from action bounds and take $A=\mathbb{R}$ to make the relationship between norm perceptions, actions, and empirical expectations analytically transparent. Individual $i$'s optimal action is
\begin{equation}\label{eq:UCE_action}
    a_i= \mathbb{E}[N_{-i}\mid H_i]-\frac{1}{2\theta}.
\end{equation}
A change in individual $i$'s perceived social norm translates one-for-one into her action, so the message-form and disclosure-mode rankings derived below also apply to behavioral responses in this benchmark.

We next define empirical expectations. An empirical expectation, sometimes referred to as a descriptive norm, concerns what others are expected to do. It differs from a perceived injunctive social norm: the former concerns anticipated behavior, whereas the latter concerns what an individual believes others regard as appropriate. Existing evidence suggests that the two are often correlated (see, e.g., \citealt{Thogersen2008,KolleQuercia2021,PrzepiorkaSzekelyAndrighetto2022}).
Our framework defines individual $i$'s \textbf{empirical expectation} as 
\begin{equation*}
    \mathbb{E}\!\left[\mathrm{avg}_{j\in I\setminus \{i\}} a_j \mid H_i\right],
\end{equation*}
which is $i$'s expectation of the average action taken by the other group members. Since $a_j$ is affected by $\mathbb{E}[N_{-j}\mid H_j]$, this expectation inherits a dependence on the perceived social norms that enter others' decisions.

To make this connection explicit, return to the MI case. Using the linear response in the UCE benchmark, we obtain the following relationship:
\begin{equation}
    \begin{aligned}
        \mathbb{E}\!\left[\mathrm{avg}_{j\in I\setminus \{i\}} a_j \mid y_i\right] 
        &= \mathbb{E}\!\left[\mathrm{avg}_{j\in I\setminus \{i\}} \mathbb{E}[N_{-j}|y_j] \mid y_i\right]-\frac{1}{2\theta}\\
        &=\frac{\nu_{\epsilon}}{\nu_{\epsilon}+\nu_s} \mu_s + \frac{\nu_s}{\nu_{\epsilon}+\nu_s} \mathbb{E}[\mathrm{avg}_{j\in I\setminus\{i\}}\mathbb{E}[S|y_j]| y_i] -\frac{1}{2\theta}\\
        &=\frac{\nu_{\epsilon}}{\nu_{\epsilon}+\nu_s} \mu_s + \frac{\nu_s}{\nu_{\epsilon}+\nu_s} \mathbb{E}[N_{-i}|y_i] -\frac{1}{2\theta}.
    \end{aligned}
\end{equation}
This expression shows how perceived social norms and empirical expectations, while conceptually distinct, can be systematically linked when behavior responds to perceived norms. The coefficient $\nu_s/(\nu_{\epsilon}+\nu_s)$, which is less than $1$, summarizes the responsiveness of an individual's empirical expectation of others' average action to her perceived social norm.\footnote{The closed-form coefficient follows from the maintained Gaussian prior-cue structure.} This response is weaker than that of her own action in this benchmark because anticipating others' actions also requires inferring how they form their norm perceptions from private cues. The coefficient increases when private cues receive more weight in beliefs about appropriateness.

This link also clarifies an important measurement distinction: elicited empirical expectations are belief measures, not substitutes for realized actions. 
\begin{remark} In the MI case, under the UCE benchmark,
\begin{equation*}
    \mathrm{avg}_{i\in I}\left(
\mathbb{E}\!\left[
\mathrm{avg}_{j\in I\setminus\{i\}} a_j
\,\middle|\, y_i
\right]
\right)
=
\frac{\nu_\epsilon}{\nu_\epsilon+\nu_s}\mu_s
+
\frac{\nu_s}{\nu_\epsilon+\nu_s}
\mathrm{avg}_{i\in I}\mathbb{E}[N_{-i}\mid y_i]
-\frac{1}{2\theta},
\end{equation*}
whereas \begin{equation*}
    \mathrm{avg}_{i\in I} \left(a_i\right)=\mathrm{avg}_{i\in I}\mathbb{E}[N_{-i}\mid y_i]
-\frac{1}{2\theta}.
\end{equation*}
\end{remark}
Thus, the average elicited empirical expectation is generally not the same as the realized average action. This distinction also matters for social-information messages: for example, if individuals differ in the importance they attach to conforming to perceived social norms, previous actions and elicited empirical expectations need not convey the same information about others' perceived social norms.

The mechanism is not specific to the UCE benchmark. More generally, when behavior is shaped by material incentives and norm compliance, expectations about others' actions are connected to beliefs about their perceived social norms. Furthermore, since this inference is imperfect, subjective expectations need not coincide with the realized group average action.

By linking perceived social norms to actions and empirical expectations, the framework clarifies how behavior can convey indirect information about norm perceptions. The following sections focus on messages reporting previous participants' personal values and perceived social norms, examining how their effects on current perceived norms depend on both the disclosed statistic and the disclosure mode.


\section{Message Forms}\label{sec:message_forms}

Across settings ranging from intergroup relations and fairness to political expression, a broad literature examines how social norms and their perceptions shape behavior and inform behavioral interventions \citep{Paluck2009,PaluckGreen2009,TankardPaluck2016,BicchieriChavez2010,BursztynEgorovFiorin2020}. For interventions that communicate norm-relevant information, an important design question is how the form of that information influences perceived social norms and subsequent behavior. To examine this question, we consider a setting in which individuals receive norm-relevant information collected from a separate previous group while forming perceived norms about their current reference group. Holding the relevant situation fixed, we assume that both groups face the same underlying appropriateness standard $S$. Thus, the disclosed information affects perceived norms through what it implies about the underlying appropriateness standard, rather than through direct observation of reference-group assessments or actions.

Let $K$ denote the set of individuals in the previous group, and let $x_k$ denote the report collected from previous-group member $k\in K$. We assume that previous-group reports are elicited under MI, with each member observing only her own private cue. For each $k\in K$, let $N_{-k}:=\mathrm{avg}_{\ell\in K \setminus \{k\}} \mathbb{E}[S\mid y_{\ell}]$, where $|K|\geq 2$. In the analysis below, previous-group information is not part of the reference-group norm the recipient is trying to infer; its relevance is informational, through what it implies about the underlying appropriateness standard. We focus on two message forms: previous participants' elicited personal values and their elicited perceived social norms. Both are norm-relevant belief measures, but their reporting forms can affect how recipients respond to changes in the disclosed statistics. The framework makes this distinction explicit by showing how different message forms summarize previous participants' private cues and allowing their effects on the recipient's perceived social norm to be compared.

To focus on the recipient's own updating, this section considers private disclosure. Individual $i$ is the recipient: she privately observes $(x_k)_{k\in K}$, while the other members of her current reference group observe only their own private cues.\footnote{That is, $H_i=\{\left(y_i,(x_k)_{k\in K}\right)\}$ and $H_j=\{y_j\}$ for all $j\in I\setminus\{i\}$.} These privately observed reports therefore affect $i$'s perceived social norm through her own inference about the underlying appropriateness standard, rather than through any update by her current reference-group members.

\paragraph{Previous-group cue decomposition}  

We first decompose disclosure effects through previous-group cues, before turning to how elicited message forms summarize that information. For this analytical decomposition, let $H_i=\{\left(y_i,(y_k)_{k\in K}\right)\}$, while $H_j=\{y_j\}$ for each $j\in I\setminus \{i\}$. Individual $i$'s perceived social norm then satisfies 
\begin{equation}\label{eq:private_main}
    \mathbb{E}\!\left[N_{-i} \middle| y_i,(y_k)_{k\in K}\right] 
    =
\frac{\nu_{\epsilon}}{\nu_s+\nu_{\epsilon}} \mu_s
+
\frac{\nu_s}{\nu_s+\nu_{\epsilon}}
\mathbb{E}[S \mid y_i, (y_k)_{k\in K}].
\end{equation}
The structure is parallel to Equation~\eqref{eq:perceived_norm_ini_value} in the MI case. Since other current-group members still condition only on their own cues, previous-group cues affect $i$'s perceived social norm only through her updated assessment of $S$.

Let \begin{equation*}
    \bar{y}_K:=\mathrm{avg}_{k\in K} y_k 
\end{equation*}
denote the average previous-group cue. Under the framework's information structure, $\bar{y}_K$ is sufficient for the previous-group cue realizations in $i$'s inference about $S$. Substituting 
\begin{equation}\label{eq:assessment_i_k}
    \mathbb{E}[S \mid y_i, (y_k)_{k\in K}]=\frac{
    \nu_\epsilon\mu_s
    +
    \nu_s y_i
    +
    |K|\nu_s\bar{y}_K
}{
    \nu_\epsilon+(|K|+1)\nu_s},
\end{equation}
into \eqref{eq:private_main} yields
\begin{equation}\label{eq:private_bench}
    \mathbb{E}\!\left[N_{-i}\middle| y_i,(y_k)_{k\in K}\right] =\begin{aligned}
        \frac{
\nu_\epsilon\big(\nu_\epsilon+(|K|+2)\nu_s\big)\mu_s
+
\nu_s^2 y_i
+
|K|\nu_s^2\bar{y}_K
}{
(\nu_s+\nu_\epsilon)\big(\nu_\epsilon+(|K|+1)\nu_s\big)
} \text{.}
    \end{aligned}
\end{equation}
Thus, individual $i$'s perceived social norm is a weighted average of the prior mean $\mu_s$, her own cue $y_i$, and the previous-group average cue $\bar{y}_K$.

The comparative statics of this cue-decomposition effect are intuitive. We refer to the marginal effect of $\bar {y}_K$ in \eqref{eq:private_bench} as the \textbf{cue-decomposition effect}: it captures how previous-group cue information would enter perceived-norm updating before accounting for how that information is encoded in a disclosed statistic. This effect increases with the number of previous-group cues, $|K|$. It also increases with appropriateness uncertainty, $\nu_s$, and decreases with cue noise, $\nu_\epsilon$, because these changes make private cues more informative relative to the prior mean.

Equation \eqref{eq:private_main} shows that, under private disclosure, previous-group cue information shifts $i$'s perceived social norm through her assessment of $S$. Therefore, the effect of a disclosed statistic depends on how it summarizes the previous group's information about $S$. We next turn to two such message forms: elicited personal values and elicited perceived social norms.

\paragraph{Elicited personal values}

We then analyze the message form that reports previous participants' elicited personal values. Each reported value represents a previous participant's own assessment of $S$, formed from their private cues.

Denote the average of these \textbf{elicited personal values} by
\begin{equation*}
    \bar{r}_K:=\mathrm{avg}_{k\in K}\mathbb{E}[S\mid y_k].
\end{equation*}
The statistic $\bar{r}_K$ summarizes $\bar {y}_K$ through the mapping from private cues into reported personal values. We refer to the mapping from cues into a disclosed statistic as the \textbf{statistic-encoding effect}. Using the relation \eqref{eq:value_cue_init}, individual $i$ can infer the average private cue from $\bar{r}_K$:
\begin{equation}\label{eq:k_value_to_cue}
    \bar{y}_K= \frac{(\nu_s+\nu_{\epsilon})}{\nu_s} \bar{r}_K - \frac{\nu_{\epsilon}}{\nu_s} \mu_s.
\end{equation}
Thus, within the framework, $\bar{r}_K$ is an encoded summary of $\bar{y}_K$.
Together with \eqref{eq:private_bench}, this implies that, after privately observing $\bar{r}_K$, $i$'s perceived social norm is
\begin{equation}\label{eq:private_k_per_value}
    \mathbb{E}\!\left[
    N_{-i}
    \middle| y_i,\bar{r}_K
    \right]
    =
    \frac{
    \nu_\epsilon(\nu_\epsilon+2\nu_s)\mu_s
    +
    \nu_s^2 y_i
    +
    |K|\nu_s(\nu_s+\nu_\epsilon)\bar{r}_K
    }{
    (\nu_s+\nu_\epsilon)\big(\nu_\epsilon+(|K|+1)\nu_s\big)
    }.
\end{equation}
Thus, individual $i$'s perceived social norm is a weighted average of the prior mean $\mu_s$, her own cue $y_i$, and the previous-group average elicited personal value $\bar{r}_K$. The average $\bar{r}_K$ therefore summarizes the relevant previous-group information entering $i$'s updating. 

In sum, under private disclosure, $\bar{r}_K$ acts as an additional signal about $S$. Its weight combines the cue-decomposition effect of $\bar{y}_K$ with the statistic-encoding effect through which $\bar{r}_K$ encodes $\bar{y}_K$. This combined effect follows the same comparative statics as the cue-decomposition effect: it increases with $|K|$, increases with appropriateness uncertainty $\nu_s$, and decreases with cue noise $\nu_{\epsilon}$. These changes make the previous-group average elicited personal value more informative about the underlying standard.

\paragraph{Elicited perceived social norms}

We next analyze the message form that reports previous participants' elicited perceived social norms. These reports capture what previous participants believe other members of their own reference group regard as appropriate, based on their private cues.

Denote the average of these \textbf{elicited perceived social norms} by
\begin{equation*}
    \hat{N}_K:=\mathrm{avg}_{k\in K}\big(\mathbb{E}[N_{-k}\mid y_k]\big).
\end{equation*}

The statistic $\hat{N}_K$ is an encoded summary of $\bar{y}_K$, but through a different encoding from $\bar {r}_K$. Using the relation in \eqref{eq:perceived_norm_full_MI}, individual $i$ can infer the average previous-group cue from $\hat{N}_K$:
\begin{equation}\label{eq:k_norm_to_cue}
    \bar{y}_K
    =
    \left(\frac{\nu_s+\nu_\epsilon}{\nu_s}\right)^2
    \hat{N}_K
    -
    \frac{\nu_\epsilon(\nu_\epsilon+2\nu_s)}{\nu_s^2}\mu_s.
\end{equation}
The negative coefficient on $\mu_s$ reflects the need to remove the prior component embedded in perceived-norm reports when recovering $\bar{y}_K$. Previous participants form these reports by combining their private-cue information with the prior mean, so $\hat{N}_K$ reflects the previous group's cue information through reports already pulled toward the prior. Recovering $\bar{y}_K$ therefore requires correcting for this prior pull.

Substituting \eqref{eq:k_norm_to_cue} into \eqref{eq:private_bench} yields
 \begin{equation}\label{eq:private_k_norm}
    \mathbb{E}\!\left[N_{-i}\middle| y_i,\hat{N}_K\right]
    =
    \frac{
    \nu_\epsilon\big((1-|K|)\nu_\epsilon+(2-|K|)\nu_s\big)\mu_s
    +
    \nu_s^2 y_i
    +
    |K|(\nu_s+\nu_\epsilon)^2\hat{N}_K
    }{
    (\nu_s+\nu_\epsilon)\big(\nu_\epsilon+(|K|+1)\nu_s\big)
    }.
\end{equation}
The coefficient on $\hat{N}_K$ is 
\begin{equation*}
     \frac{|K|(\nu_s+\nu_\epsilon)}
    {\nu_\epsilon+(|K|+1)\nu_s},
\end{equation*}
which decreases in $\nu_s$ and increases in $\nu_{\epsilon}$.

Unlike \eqref{eq:private_bench}, this expression need not be a weighted average of $\mu_s$, $y_i$, and $\hat{N}_K$: for $|K|\geq 2$, the coefficient on $\mu_s$ is negative. This negative coefficient carries over from the encoding correction in \eqref{eq:k_norm_to_cue}: $\hat{N}_K$ is informative about the previous group's cues only after correcting for the prior pull embedded in the perceived-norm report.

Observed previous actions fit into this case as an indirect norm-relevant statistic. For example, in the UCE benchmark, each previous action satisfies $a_k=\mathbb{E}[N_{-k}\mid y_k]-1/(2\theta)$, so $\bar{a}_K=\hat{N}_K-1/(2\theta)$. Thus, observing $\bar{a}_K$ allows individual $i$ to infer the corresponding $\hat{N}_K$, and then infer the implied $\bar{y}_K$ through the encoding relation in \eqref{eq:k_norm_to_cue}.

In sum, under private disclosure, the effect of $\hat{N}_K$ depends on how this statistic encodes the previous-group cue information. The coefficient on $\hat{N}_K$ combines two forces. The cue-decomposition effect works through the informativeness of the previous group's cues about $S$; by itself, it is stronger when $\nu_s$ is larger or $\nu_{\epsilon}$ is smaller. The statistic-encoding effect works in the opposite direction. When $\nu_s$ is larger or $\nu_{\epsilon}$ is smaller, previous participants place more weight on their private cues when forming perceived-norm reports, so $\hat {N}_K$ is pulled less toward $\mu_s$. A given movement in $\hat{N}_K$ then corresponds to a smaller movement in the underlying $\bar{y}_K$, so individual $i$ places \emph{lower} weight on the statistic. Since the statistic-encoding effect dominates, the comparative statics are reversed.\footnote{In a macroeconomic setting, \cite{HuoPedroni2023} make a related observation: the effect of new information depends on how strongly it encodes the private cues.}

\paragraph{Comparing the two message forms}

Under private disclosure, when $\hat{N}_K$ and $\bar{r}_K$ are reported on the same scale, comparing \eqref{eq:private_k_per_value} and \eqref{eq:private_k_norm} gives the first main result.
\begin{proposition}
    Under private disclosure, when $\hat{N}_K$ and $\bar{r}_K$ are reported on the same scale, a given numerical change in $\hat{N}_K$ induces a larger change in individual $i$'s perceived social norm than an equally sized change in $\bar{r}_K$.
\end{proposition}

\begin{proof}
By \eqref{eq:private_k_per_value}, the coefficient on $\bar{r}_K$ is
\begin{equation*}
    \frac{|K|\nu_s}{\nu_\epsilon+(|K|+1)\nu_s}.
\end{equation*}
By \eqref{eq:private_k_norm}, the coefficient on $\hat{N}_K$ is
\begin{equation*}
    \frac{|K|(\nu_s+\nu_\epsilon)}
    {\nu_\epsilon+(|K|+1)\nu_s}.
\end{equation*}
Both coefficients are positive, and the ratio of the former to the latter is
\begin{equation*}
    \frac{\nu_s}{\nu_s+\nu_\epsilon}<1.
\end{equation*}
\end{proof}

This difference stems from the different statistic-encoding effects of $\bar{r}_K$ and $\hat{N}_K$. Although these two statistics convey equivalent information about the previous-group average cue, they encode that information differently.\footnote{For related models in which responses depend on how underlying information is encoded, see \cite{MooreHealy2008,FrydmanJin2022,OpreaVieider2024}.} Perceived-norm reports are pulled more strongly toward the common prior than personal-value reports because previous participants infer others' assessments from their own. Recipients account for this stronger prior pull when interpreting perceived-norm reports. Consequently, a given numerical change in $\hat{N}_K$ corresponds to a larger inferred change in $\bar{y}_K$ than an equally sized change in $\bar{r}_K$, and therefore induces a larger adjustment in individual $i$'s perceived social norm. Therefore, this distinction matters when comparing interventions based on different elicited normative measures: response coefficients for personal-value messages cannot be directly transferred to
perceived-norm messages, even when both statistics are reported on the same scale.

Moreover, the two message forms differ in how their marginal effects vary with relative cue noise. Let $t:=\nu_\epsilon/\nu_s$, and denote the coefficients on $\bar r_K$ and $\hat N_K$ by $\lambda_r$ and $\lambda_N$, respectively. From the expressions in the proof,
\begin{equation*}
\frac{\partial\lambda_r}{\partial t}
=-\frac{|K|}{(t+|K|+1)^2}<0,
\qquad
\frac{\partial\lambda_N}{\partial t}
=\frac{|K|^2}{(t+|K|+1)^2}>0.
\end{equation*}
Thus, as relative cue noise increases, the marginal response to elicited personal values decreases, whereas that to elicited perceived social norms increases, widening the difference between the two message forms.

\section{Disclosure Mode}\label{sec:disclosure_mode}

Experiments have directly compared private and public dissemination of norm-relevant information. For example, \cite{Arias2019} examines private versus public transmission of a radio program addressing violence against women, while \cite{Luckner2026} compares private and public dissemination of a fairness-norm message based on previous participants' judgments in a public-goods experiment. Their findings illustrate that disclosure mode can have consequential, context-dependent effects on behavior. Our framework clarifies a belief mechanism through which disclosure mode can shape perceived social norms.

In this section, we hold the message form fixed and vary only the disclosure mode. We continue to let $x_k$ denote the report collected from the previous-group member $k\in K$. Under public disclosure, all current-group members observe the same set of reports $(x_k)_{k\in K}$, so $H_i=\{\left(y_i,(x_k)_{k\in K}\right)\}$ and $H_j=\{\left(y_j,(x_k)_{k\in K}\right)\}$ for all $j\in I\setminus\{i\}$. Public disclosure therefore affects individual $i$'s perceived social norm not only through her own inference about $S$, but also through her belief that other current-group members update their assessments using the same reports.

\paragraph{Previous-group cue decomposition} 

To see how public disclosure differs from private disclosure, we decompose disclosure effects through previous-group cues before turning to how elicited messages summarize that information. For this analytical decomposition, take $H_i=\{\left(y_i,(y_k)_{k\in K}\right)\}$ and $H_j=\{\left(y_j,(y_k)_{k\in K}\right)\}$ for each $j\in I\setminus\{i\}$.

Recall that $\bar{y}_K=\mathrm{avg}_{k\in K} y_k$. The previous-group cue realizations enter the updating formula through this average. Applying the same updating formula as in \eqref{eq:assessment_i_k} to each current-group member $j$ gives
\begin{equation*}
    \mathbb{E}[S\mid y_j, (y_k)_{k\in K}]
=
\frac{\nu_\epsilon}{\nu_\epsilon+(|K|+1)\nu_s}\mu_s
+
\frac{\nu_s}{\nu_\epsilon+(|K|+1)\nu_s}y_j
+
\frac{|K|\nu_s}{\nu_\epsilon+(|K|+1)\nu_s}\bar y_K .
\end{equation*}

Thus, $i$'s expectation of $j$'s updated assessment is 
\begin{equation*}
    \begin{aligned}
        &\mathbb{E}\!\left[
        \mathbb{E}[S\mid y_j,(y_k)_{k\in K}]
        \middle|
        y_i,(y_k)_{k\in K}
        \right] =
        \frac{\nu_\epsilon}{\nu_\epsilon+(|K|+1)\nu_s}\mu_s \\
        &\quad+
        \frac{\nu_s}{\nu_\epsilon+(|K|+1)\nu_s}
        \mathbb{E}[y_j\mid y_i,(y_k)_{k\in K}] +
        \frac{|K|\nu_s}{\nu_\epsilon+(|K|+1)\nu_s}\bar y_K .
    \end{aligned}
\end{equation*}

Since $\mathbb{E}[y_j|y_i, (y_k)_{k\in K}]=\mathbb{E}[S|y_i, (y_k)_{k\in K}]$, previous-group cues also shift $i$'s expectation of $y_j$ through her updated assessment of $S$. Using symmetry across $j\in I\setminus\{i\}$ and averaging gives
\begin{equation}\label{eq:public_bench}
    \begin{aligned}
        &\mathbb{E}\!\left[\mathrm{avg}_{j\in I\setminus \{i\}}\Big(\mathbb{E}[S\mid y_j,(y_k)_{k\in K}]\Big)\middle| y_i,(y_k)_{k\in K}\right] \\
    &=
    \left(\frac{\nu_s}{\nu_{\epsilon}+(|K|+1)\nu_s}\right)^2 y_i
    +\frac{\nu_{\epsilon}}{\nu_{\epsilon}+(|K|+1)\nu_s}
    \left(1+\frac{\nu_s}{\nu_{\epsilon}+(|K|+1)\nu_s}\right)\mu_s \\
    &\quad
    +\frac{|K|\nu_s}{\nu_{\epsilon}+(|K|+1)\nu_s}
    \left(1+\frac{\nu_s}{\nu_{\epsilon}+(|K|+1)\nu_s}\right)\bar{y}_K.
    \end{aligned}
\end{equation}
Thus, individual $i$'s perceived social norm is a weighted average of the prior mean $\mu_s$, her own cue $y_i$, and the previous-group average cue $\bar{y}_K$. 
The cue-decomposition effect, which is the coefficient on $\bar{y}_K$,
\begin{equation*}
    \frac{|K|\nu_s}{\nu_{\epsilon}+(|K|+1)\nu_s}
    \left(1+\frac{\nu_s}{\nu_{\epsilon}+(|K|+1)\nu_s}\right),
\end{equation*}
increases as $|K|$ increases, $\nu_s$ increases, and $\nu_{\epsilon}$ decreases.

Different from private disclosure, the public cue-decomposition effect combines two forces. The first is a direct public-updating channel: previous-group cue information enters other current-group members' assessments of $S$. The second is the inference channel also present under private disclosure: the same cue information shifts $i$'s belief about other current-group members' private cues, $\mathbb{E}[y_j\mid y_i,(y_k)_{k\in K}]$, through $i$'s updated assessment of $S$. This additional direct-updating channel helps explain why public disclosure has a stronger effect than private disclosure, a result established later in this section.

\paragraph{Disclosed message forms}
 
We now turn to the disclosed message forms. For the previous-group average personal value, $\bar{r}_K$, Section~\ref{sec:message_forms} provides the inverse mapping \eqref{eq:k_value_to_cue}:
\begin{equation*}
    \bar{y}_K= \frac{(\nu_s+\nu_{\epsilon})}{\nu_s} \bar{r}_K - \frac{\nu_{\epsilon}}{\nu_s} \mu_s.
\end{equation*}
Substituting this relation into \eqref{eq:public_bench} shows that, when $\bar{r}_K$ is publicly disclosed, individual $i$'s perceived social norm is
\begin{equation}\label{eq:public_k_per_value}
    \begin{aligned}
        &\mathbb{E}\!\left[\mathrm{avg}_{j \in I\setminus \{i\}}\!\left(\mathbb{E}[S\mid y_j,\bar{r}_K]\right)\middle|y_i,\bar{r}_K\right] \\
&= \left(\frac{\nu_s}{\nu_{\epsilon}+(|K|+1)\nu_s}\right)^2 y_i \\
&\quad+
\frac{(1-|K|)\nu_\epsilon\big(\nu_\epsilon+(|K|+2)\nu_s\big)}
{\big(\nu_{\epsilon}+(|K|+1)\nu_s\big)^2}\,\mu_s \\
&\quad+
\frac{|K|(\nu_s+\nu_\epsilon)\big(\nu_\epsilon+(|K|+2)\nu_s\big)}
{\big(\nu_{\epsilon}+(|K|+1)\nu_s\big)^2}\,\bar{r}_K.
    \end{aligned}
\end{equation}
Unlike the cue-decomposition expression in \eqref{eq:public_bench}, \eqref{eq:public_k_per_value} need not be a weighted average of $\mu_s$, $y_i$, and $\bar r_K$: when $|K|\geq 2$, the coefficient on $\mu_s$ is negative. This follows from the statistic-encoding relation between $\bar {r}_K$ and $\bar {y}_K$. The report $\bar r_K$ already combines previous participants' cue information with the prior mean, so using it to infer  $\bar y_K$ requires correcting for the prior pull embedded in the personal-value report. This correction reverses the comparative statics with respect to $\nu_s$ and $\nu_\epsilon$: the coefficient on $\bar r_K$ decreases in $\nu_s$ and increases in $\nu_\epsilon$, opposite to the cue-decomposition effect.

For the average of previous participants' elicited perceived social norms, Section~\ref{sec:message_forms} also gives the inverse mapping \eqref{eq:k_norm_to_cue}:
\begin{equation*}
    \bar{y}_K
    =
    \left(\frac{\nu_s+\nu_\epsilon}{\nu_s}\right)^2
    \hat{N}_K
    -
    \frac{\nu_\epsilon(\nu_\epsilon+
    2\nu_s)}{\nu_s^2}\mu_s.
\end{equation*}

Substituting the relation into \eqref{eq:public_bench} gives
\begin{equation}\label{eq:public_k_norm}
    \begin{aligned}
        &\mathbb{E}\!\left[\mathrm{avg}_{j\in I\setminus \{i\}}\!\left(\mathbb{E}[S\mid y_j,\hat{N}_K]\right)\middle|y_i,\hat{N}_K\right] \\
&= \left(\frac{\nu_s}{\nu_{\epsilon}+(|K|+1)\nu_s}\right)^2 y_i \\
&\quad+
\Bigg[
    \frac{\nu_{\epsilon}}{\nu_{\epsilon}+(|K|+1)\nu_s}
    \left(1+\frac{\nu_s}{\nu_{\epsilon}+(|K|+1)\nu_s}\right) \\
&\qquad\qquad
    -\frac{|K|\nu_s}{\nu_{\epsilon}+(|K|+1)\nu_s}
    \left(1+\frac{\nu_s}{\nu_{\epsilon}+(|K|+1)\nu_s}\right)
    \frac{\nu_\epsilon(\nu_\epsilon+2\nu_s)}{\nu_s^2}
\Bigg]\mu_s \\
&\quad+
\frac{|K|\nu_s}{\nu_{\epsilon}+(|K|+1)\nu_s}
\left(1+\frac{\nu_s}{\nu_{\epsilon}+(|K|+1)\nu_s}\right)
\frac{(\nu_s+\nu_\epsilon)^2}{\nu_s^2}
\,\hat{N}_K.
    \end{aligned}
\end{equation}
The coefficient on $\hat{N}_K$ is 
\begin{equation*}
    \frac{|K|(\nu_s+\nu_\epsilon)^2\big(\nu_\epsilon+(|K|+2)\nu_s\big)}
{\nu_s\big(\nu_{\epsilon}+(|K|+1)\nu_s\big)^2},
\end{equation*}
which is positive, increases in $|K|$, decreases in $\nu_s$, and increases in $\nu_{\epsilon}$. The reversal of the comparative statics with respect to $\nu_s$ and $\nu_\epsilon$, relative to the cue-decomposition effect, reflects the statistic-encoding effect.


Observed previous actions can be treated as an indirect version of the second message form. For example, in the UCE benchmark, $\bar a_K=\hat N_K-1/(2\theta)$, so a publicly disclosed average action is equivalent, up to a known shift, to a publicly disclosed average perceived-norm report. Its effect therefore follows from the public-disclosure effect of $\hat {N}_K$.

\paragraph{Comparing disclosure modes}

The public-disclosure expressions above show that the coefficient on the disclosed statistic is larger than its private-disclosure counterpart for both $\bar r_K$ and $\hat N_K$. 

\begin{proposition}
For either disclosed statistic, $\bar r_K$ or $\hat N_K$, the marginal effect on individual $i$'s perceived social norm is larger under public disclosure than under private disclosure.
\end{proposition}

The comparison isolates the belief-updating channel of public disclosure. The source of the difference is the cue-decomposition effect. A given statistic encodes $\bar{y}_K$ in the same way whether it is disclosed privately or publicly. What changes is how $\bar y_K$ enters perceived-norm updating. Under public disclosure, previous-group cue information affects not only $i$'s assessment of $S$, but also the assessments that $i$ expects other current-group members to form. This additional channel strengthens the cue-decomposition effect under public disclosure, thereby increasing the marginal response to either disclosed statistic.\footnote{The proportional increase in the marginal response under public disclosure is the same for both aforementioned message forms, because disclosure mode changes the cue-decomposition effect while leaving the statistic-encoding effect unchanged.} 

\begin{proof}
The coefficient on $\bar{y}_K$ in \eqref{eq:public_bench} is
\begin{equation*}
    \frac{|K|\nu_s}{\nu_{\epsilon}+(|K|+1)\nu_s}
    \left(1+\frac{\nu_s}{\nu_{\epsilon}+(|K|+1)\nu_s}\right),
\end{equation*}
whereas the coefficient on $\bar{y}_K$ in \eqref{eq:private_bench} is
\begin{equation*}
    \frac{ |K|\nu_s^2 }{ (\nu_s+\nu_\epsilon)\big(\nu_\epsilon+(|K|+1)\nu_s\big)
}.
\end{equation*}
Both coefficients are positive. Their ratio is
\begin{equation*}
    \frac{\nu_{\epsilon}+\nu_s}{\nu_s}\left(1+\frac{\nu_s}{\nu_{\epsilon}+(|K|+1)\nu_s}\right)>1,
\end{equation*}
since $\nu_s>0$, $\nu_\epsilon>0$, and $|K|\geq 2$. Thus, the cue-decomposition effect is stronger under public disclosure. Because the mappings from $\bar r_K$ and $\hat N_K$ to $\bar y_K$ have positive slopes and do not depend on the disclosure mode, the same inequality applies to the marginal effects of both disclosed statistics.
\end{proof}
\bigskip

This result shows why disclosure mode matters even when the disclosed statistic is held fixed. For a given private cue and disclosed statistic, individual $i$ forms the same assessment of $S$ under either disclosure mode; the difference arises from her expectations about other current-group members' assessments. This distinction connects the result to expectation-based accounts of social norms \citep{Bicchieri2006} and to field evidence that mutual knowledge can strengthen the behavioral effect of social-norm information \citep{HagerEtAl2025}.\footnote{\cite{HagerEtAl2025} find that women's attendance at community meetings in Kyrgyzstan responds more strongly to information about societal support for women's political participation when recipients are told that other community members also receive the information.} For social-information interventions, this implies that the choice of whom to inform, and whether recipients know that others receive the same information, can affect perceived-norm responses even when the reported statistic remains unchanged.


\section{Conclusion}\label{sec:conclusion}


Social-information messages are widely used to influence perceived norms and related behaviors. These messages differ both in what they report and in how they are disclosed. The framework developed here clarifies why these distinctions matter: the response of perceived social norms to a reported statistic depends on both its message form and its disclosure mode.


The analysis yields two formal implications. First, holding the disclosure mode fixed, a given numerical change in disclosed perceived social norms induces a larger change in current perceived norms than an equally sized change in disclosed personal values.\footnote{The message-form ordering established under private disclosure also holds under public disclosure.} Second, holding the disclosed statistic fixed, the marginal effect of that statistic on current perceived norms is larger under public disclosure than under private disclosure. The belief-based framework not only makes these comparisons possible, but also reveals the mechanisms behind them. Elicited personal values and perceived social norms encode previous participants' private cue information in distinct ways, generating different statistic-encoding effects across message forms. Public disclosure, relative to private disclosure, adds a further updating channel: the recipient expects other members of her reference group to also incorporate the same disclosed statistic into their own assessments, thereby strengthening the cue-decomposition effect.


As a broad implication, the effect of a social-information intervention can depend on what statistic the message reports and how it is disclosed, not only on context or population. For intervention design, the choice of which normative assessment to report matters for how recipients interpret and respond to changes in the reported statistic. Moreover, disclosure mode is more than a matter of delivery: public disclosure allows recipients to anticipate how other members of their reference group update their assessments using the same information, introducing an additional channel absent under private disclosure. Thus, the effect of a norm-relevant message depends not only on the value it reports, but also on the belief process that produced it and on who is known to observe it. Designers should therefore consider and make explicit both how a message is constructed, including which normative assessments are elicited and summarized in the reported statistic, and how it is disclosed, including who receives it and whether recipients know that others receive the same information.

\clearpage
\bibliographystyle{chicago}
\bibliography{ref}

\end{document}